\documentclass{jpp}
\usepackage{graphicx}

\usepackage[utf8]{inputenc}
\usepackage[T1]{fontenc}
\usepackage{amsmath}
\usepackage[normalem]{ulem}

\newtheorem{theorem}{Theorem}
\newtheorem{lemma}{Lemma}

\renewcommand{\d}{\mathrm d}
\renewcommand{\vec}[1]{\boldsymbol{#1}}
\newcommand{\R}{\mathbb R}
\newcommand{\norm}[1]{\left\| #1 \right\|}
\newcommand{\grad}{\boldsymbol\nabla}
\newcommand{\cross}{\times}
\newcommand{\pder}[2]{\frac{\partial #1}{\partial #2}}

\shorttitle{Fast quadrature in axisymmetry}
\shortauthor{E. Toler, A.J. Cerfon, and D. Malhotra}

\title{A fast, accurate, and easy to implement Kapur-Rokhlin quadrature scheme for singular integrals in axisymmetric geometries}

\author{Evan Toler\aff{1}
  \corresp{\email{evan.toler@cims.nyu.edu}},
  A. J. Cerfon\aff{1}
 \and D. Malhotra\aff{2}}

\affiliation{\aff{1}Courant Institute of Mathematical Sciences, New York University, New York, NY 10012\\
\aff{2}Flatiron Institute, New York, NY 10012}

\begin{document}

\maketitle

\begin{abstract}
Many applications in magnetic confinement fusion require the efficient calculation of surface integrals with singular integrands. The singularity subtraction approaches typically used to handle such singularities are complicated to implement and low order accurate. In contrast, we demonstrate that the Kapur-Rokhlin quadrature scheme is well-suited for the logarithmically singular integrals encountered for a toroidally axisymmetric confinement system, is easy to implement, and is high order accurate. As an illustration, we show how to apply this quadrature scheme for the efficient and accurate calculation of the normal component of the magnetic field due to the plasma current on the plasma boundary, via the virtual casing principle.
\end{abstract}

\section{Introduction}
Integral formulations and integral equations are effective and popular tools for magnetostatic and magnetohydrodynamic problems in magnetic confinement fusion \citep{Shafranov_1972,Zakharov_1973,Freidberg1976,MerkelJCP,HirshmanNEMEC,Hirshman1986,chance1997vacuum,Ludwig_2006,Ludwig_2013,Lazerson_2013,Drevlak_2018,o2018integral,malhotra2019taylor,pustovitov2021analytical}. They have intuitive physical interpretations \citep{Shafranov_1972,Zakharov_1973,Hirshman1986,Lazerson_2013,hanson2015virtual,pustovitov2021analytical}, provide geometric flexibility \citep{MerkelJCP,HirshmanNEMEC,chance1997vacuum,o2018integral,malhotra2019taylor}, and often reduce the dimension of the unknown quantities one solves for, thus reducing the number of unknowns \citep{MerkelJCP,HirshmanNEMEC,chance1997vacuum,o2018integral,malhotra2019taylor}. However, there typically is a price to pay for these advantages. Integral formulations often involve singular integrands, which are subtle to handle numerically \citep{Freidberg1976,MerkelJCP,atkinson_1997,chance1997vacuum,Ludwig_2006,Ludwig_2013,klockner2013quadrature,kress2014linear,ricketson2016accurate,malhotra2019taylor,landremanboozer}. The numerical difficulty of integrating these singular integrands depends on the nature of the singularity, the distribution of sources, and the relative location of the evaluation points (often known as target points or observation points) with respect to the sources. In this article, we will focus on the common situation in which we are trying to evaluate layer potentials at the source locations. This is for example the standard situation when applying Green's identity \citep{Freidberg1976,MerkelJCP,HirshmanNEMEC,pustovitov2008,lee2015tokamak,malhotra2019efficient}.

In the fusion community, the numerical difficulty due to the singularity of the integrand is usually addressed via the method of singularity subtraction \citep{Freidberg1976,MerkelJCP,chance1997vacuum,Ludwig_2006,Ludwig_2013}. The method is robust, but leads to a quadrature scheme with low order convergence. Furthermore, it is complicated to implement, and the chances of making mistakes in the derivation of the quadrature scheme or its numerical implementation are high. The purpose of our work is to demonstrate that for the singular layer potential integrals encountered in axisymmetric confinement devices, which can be reduced to line integrals of singular periodic functions, the Kapur-Rokhlin quadrature scheme \citep{kapur1997high} is as simple to implement as the trapezoidal rule, and is a scheme with high order convergence, leading to low error for few quadrature points. For non-axisymmetric applications, we recently presented an efficient high-order quadrature scheme based on a different approach \citep{malhotra2019efficient}, and alternative schemes may also provide good performance \citep{bruno2020,wu2021}. However, for axisymmetric cases, none of these schemes reduce to as simple and efficient a method as the Kapur-Rokhlin approach we present here.

As discussed above, there are many situations in the study of axisymmetric magnetic confinement fusion devices for which the simplicity and accuracy of the Kapur-Roklin scheme could be demonstrated. For this article, we choose to focus on the evaluation of single-layer and double-layer potentials, which are the layer potentials appearing in Green's identity, and which we will define precisely in section \ref{sec:math_background}. Our first numerical test is a numerical verification of an identity for the double-layer potential. Our second numerical test focuses on the single-layer potential, which we evaluate for an application of the virtual casing principle, to calculate the normal component of the magnetic field due to the plasma current at points on the plasma boundary \citep{Shafranov_1972,Zakharov_1973,hanson2015virtual}. 

The structure of this article is as follows. We introduce our mathematical notation, layer potential representations, and Kapur-Rokhlin quadrature for singular integrals in section \ref{sec:math_background}. In section \ref{sec:virtual_casing}, we give a brief summary of the virtual casing principle, discuss the mathematical difficulties associated with the numerical evaluation of the virtual casing integral, and describe our method for addressing these difficulties in toroidally axisymmetric geometries. We prove in section \ref{sec:analytic} that the integrands we consider in this article are logarithmically singular, and can therefore be integrated to high accuracy with the Kapur-Rokhlin quadrature scheme, which we present in section \ref{sec:KR_vecpot}. We demonstrate the accuracy and high order of convergence of the scheme for an application of the virtual casing principle and for the evaluation of a double-layer potential in section \ref{sec:numerics}, and summarize our work in section \ref{sec:conclusion}.

\section{Mathematical background}\label{sec:math_background}

\subsection{Description of toroidal volumes and surfaces}\label{sec:torsurf_desc}
Throughout our discussion of toroidal geometries, we shall make use of the standard, right-handed cylindrical coordinates $(r, \phi, z)$. At a point with toroidal angle $\phi$, we write the orthonormal unit vectors as $\vec e_r(\phi)$, $\vec e_\phi(\phi)$, and $\vec e_z$. With this notation, we emphasize the fact that the radial and azimuthal unit vectors depend on the toroidal angle. 

In this article, we will focus on axisymmetric geometries, which means that we shall only consider surfaces and volumes of revolution. We take the $z$-axis as the axis of revolution and define a simple closed curve $\gamma$ in the $(r,z)$ plane. By rotating this curve about the $z$-axis through the toroidal angle $\phi \in [0, 2\pi]$, we obtain a closed surface of revolution $\Gamma$. Its interior $\Omega$ is the corresponding volume of revolution. We refer to $\gamma$ as the generating curve of $\Gamma$. It is parameterized by a single variable $t$, which we assume has period $L$. We denote the components of $\gamma$ in the $(r,z)$ plane by $(r(t), z(t))$, and we identify a point $\vec y \in \Gamma$ by  its toroidal revolution angle $\phi$ and its generating curve parameter $t$. Correspondingly, we often write $\vec y = \vec y(\phi, t)$ to stress this parameterization. 
Moreover, we assume that $\gamma$ is a $C^1$ curve which does not intersect the $z$-axis, in the sense that the derivatives $r'(t)$ and $z'(t)$ are continuous on $[0,L]$ and there exists $R_{min}>0$ for which $r(t) \ge R_{min}$ on $[0,L]$.
Finally, we assume that $\gamma$ is oriented so that the vector $\vec n(\vec y(\phi, t)) = (\partial \vec y / \partial \phi) \cross (\partial \vec y / \partial t) / \mathcal J(t)$ is the unit outward normal to $\Omega$ at $\vec y$. The quantity $\mathcal J(t)=\norm{(\partial \vec y / \partial \phi) \cross (\partial \vec y / \partial t)}$ is the Jacobian of the parameterization.

\subsection{Single-layer and double-layer potentials for axisymmetric geometries} \label{sec:layer_pot}
Layer potentials are fundamental tools in representing solutions to the partial differential equations that arise in magnetostatic and magnetohydrodynamic calculations for magnetic confinement fusion \citep{MerkelJCP,chance1997vacuum,Ludwig_2006,Ludwig_2013,landremanboozer,Drevlak_2018}. Given a surface $\Gamma$ and a free-space Green's function $(\vec x, \vec y) \mapsto G(\vec x, \vec y)$ for a partial differential equation, the single layer operator $\mathcal S$ and the double layer operator $\mathcal D$ are defined by \citep{guenther1996partial}
\[
[\mathcal S \sigma](\vec x) = \iint_\Gamma G(\vec x, \vec y) \sigma(\vec y) \d \Gamma(\vec y)
\quad \text{and} \quad
[\mathcal D \sigma](\vec x) = \iint_\Gamma \pder{G(\vec x, \vec y)}{\vec n(\vec y)} \sigma(\vec y) \d \Gamma(\vec y),
\]
respectively. 
In the double layer representation, the quantity $\partial G(\vec x, \vec y) / \partial \vec n(\vec y) = \vec n(\vec y) \vec\cdot \grad_{\vec y} G(\vec x, \vec y)$ is the derivative of the Green's function in the outward normal direction at $\vec y$.
The function $\sigma$ is called the density function in this representation. Functions expressible as single-layer and double-layer potentials automatically satisfy the partial differential equation associated with the Green's function everywhere except the boundary $\Gamma$.

The case of Laplace's equation in three dimensions is particularly prevalent in magnetostatic and magnetohydrodynamic settings \citep{MerkelJCP,chance1997vacuum,landremanboozer,Drevlak_2018}. Here, the Green's function is
\[
G(\vec x, \vec y) = \frac 1{4 \pi \norm{\vec x - \vec y}},
\]
and functions expressible as $\mathcal S \sigma$ or $\mathcal D \sigma$ are harmonic on $\R^3 \setminus \Gamma$. Assuming the density $\sigma$ is also axisymmetric in the sense that $\partial \sigma/ \partial \phi = 0$, one can analytically compute the part of the surface integral over the revolution angle $\phi \in [0, 2\pi]$. The resulting single-layer and double-layer integrals then are one-dimensional line integrals, and the resulting Green's function in the integrand can be expressed in terms of complete elliptic integrals \citep{Ludwig_2006,Ludwig_2013,jardin2010computational}. We shall provide an explicit expression in section \ref{sec:ana_vecpot}, as we treat in detail the application of our method to the calculation of the virtual casing principle. At this point, we just highlight the fact that the singularity in the Green's function when $\vec{x}=\vec{y}$ requires the use of specialized quadrature when the target $\vec{x}$ is located on the surface $\Gamma$ , or regularization methods \citep{Freidberg1976,MerkelJCP,chance1997vacuum,landremanboozer,Drevlak_2018,malhotra2019efficient}. The purpose of this article is to show that for applications in axisymmetric geometries, the Kapur-Rokhlin quadrature scheme is simpler to implement than the known regularization methods used in the magnetic confinement community, and leads to high order convergence.

\subsection{Kapur-Rokhlin quadrature}
The Kapur-Rokhlin quadrature rules \citep{kapur1997high} are a collection of
high order schemes for computing
\[
    \int_0^b f(t) \d t
\]
when $f$ has an integrable singularity at the origin of the form $\log t$ or $t^\lambda$ for $\lambda>-1$.
We will focus on the specific Kapur-Rokhlin scheme for a
logarithmic singularity of the form
$f(t) = p(t) \log t + q(t)$, where $p$ and $q$ 
need not have known formulae, but are assumed sufficiently smooth.

\subsubsection{A Kapur-Rokhlin quadrature scheme for nonperiodic functions}
The Kapur-Rokhlin quadrature rules are corrections to the trapezoidal rule.
For the standard trapezoidal rule with equal spacing $h = b/M$, the quadrature nodes for a nonsingular integrand would be
$t_i = i h$ for $i=0, \dots, M$.
However, we omit the quadrature node $t_0=0$ since the integrand $f$ is singular there.
This yields the punctured trapezoidal rule
\[
    \int_0^b f(t) \d t \approx 
    h \left[ \sum_{i=1}^{M-1} f_i + \frac 12 f_M \right]
\]
where we have written the shorthand $f_i = f(t_i)$.
The punctured trapezoidal rule is low order accurate when $f$ is singular. For example, when $f$ is logarithmically singular at the origin, the punctured trapezoidal rule error typically decays, according to \citet{martinsson2019fast}, as $O(|h\log h|)$.

The Kapur-Rokhlin corrections place additional quadrature nodes 
outside the integration domain $[0,b]$.
Specifically, the corrections depend on a convergence rate parameter $n$ and
a smoothness parameter $m$.
One chooses $m \ge 3$ an odd integer subject to the constraint that $p$ and $q$ 
are $m$ times continuously differentiable on a wider interval $[-nh,b+mh]$. 
Under these conditions, the Kapur-Rokhlin scheme sets constants 
$\gamma_j$ for $j=\pm1, \dots, \pm n$ and
$\beta_l$ for $l = 1, \dots, (m-1)/2$ and defines the quadrature rule 
\[
    T_{m,n}^M(f) = h \left[\sum_{i=1}^{M-1} f_i + \frac 12 f_M \right]
        + h \sum_{l=1}^{(m-1)/2} \beta_l (f_{M-l} - f_{M+l})
        + h \sum_{1 \le |j| \le n} \gamma_j f_j.
\]
for $M \ge n + (m-1)/2$.
The error obeys the asymptotic rate
\[
    \left| T_{m,n}^M(f) - \int_0^b f(t) \d t \right| = \mathit O(h^n)
\] 
as $h \ttz$.
The endpoint correction coefficients $\{\gamma_j\}$ and $\{\beta_l\}$ are derived by analyzing the Euler-Maclaurin formula for quadrature error and solving a linear system for correction coefficients to obtain high-order accuracy. Figure \ref{fig:example_KR_nodes_weights} illustrates this quadrature scheme and compares it with the punctured trapezoidal rule.

\begin{figure}
    \centering
    \includegraphics[width=1\textwidth]{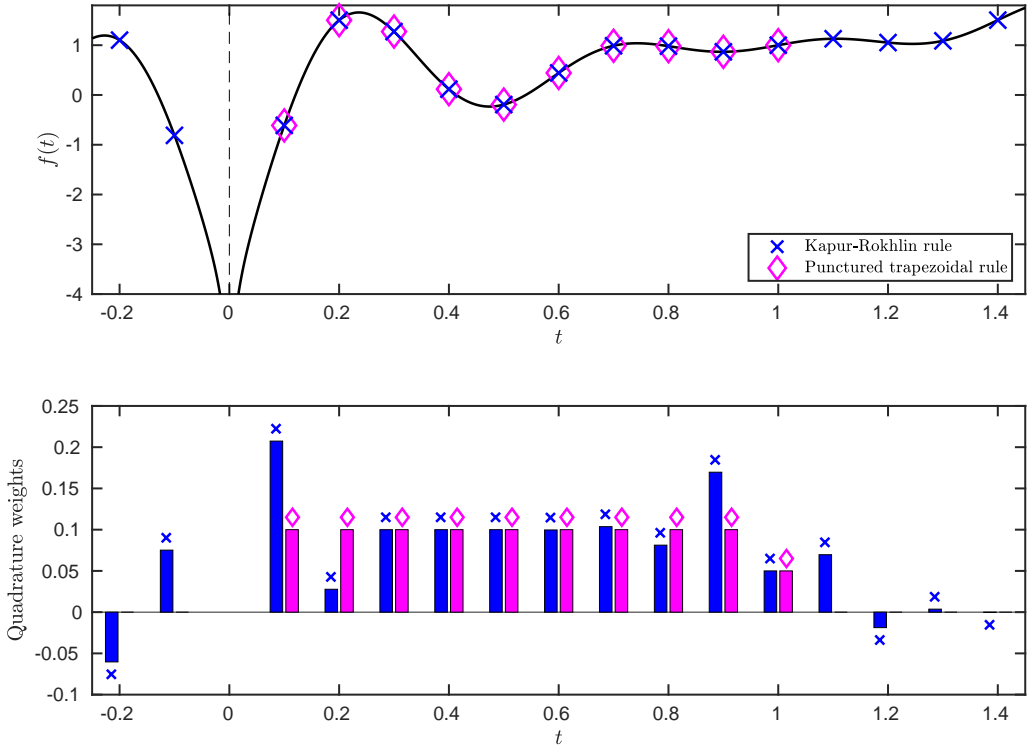}
    \caption{Nodes and weights for the Kapur-Rokhlin and punctured trapezoidal quadrature rules for estimating $\int_0^1 f(t) \d t$ with $f(t) = \cos(4\pi t) \log|t| + t$. We have used $M=10$, $n=2$, and $m=9$. Note that $\beta_4 \approx -3 \times 10^{-4}$, so the weight corrections corresponding to $t_{6}=0.6$ and $t_{14}=1.4$ are not visually discernible.}
    \label{fig:example_KR_nodes_weights}
\end{figure}

\subsubsection{A simplified quadrature for periodic functions} \label{sec:KR_periodic}
This subsection follows an argument nearly verbatim from \citet{hao2014high}.
We explain how the Kapur-Rokhlin scheme we just presented simplifies when computing
\[
    \int_{-b}^b f(t) \d t
\]
when $f$ is $2b$-periodic and logarithmically singular at the origin.
We may express these assumed properties of $f$ through the form
\[
    f(t) = p(t) \log \left|\sin \frac{\pi t}{2b} \right| + q(t),
\]
for $2b$-periodic functions $p$ and $q$.

We sum two applications of the original Kapur-Rokhlin scheme ---
one for
\[
    I_1 = \int_0^b f(t) \d t
\]
and another for
\[
    I_2 = \int_{-b}^0 f(t) \d t = \int_0^b f(-t) \d t.
\]
We assume that $p, q \in C^m[-b, b]$ and
generate $2M-1$ equispaced quadrature nodes with spacing $h=b/M$, 
defined by $t_i = ih$ for $i=\pm 1, \dots, \pm(M-1), M$.
The corrected scheme for $I_1$ is
\[
    I_1 = h \left[\sum_{i=1}^{M-1} f_i + \frac 12 f_M \right]
        + h \sum_{l=1}^{(m-1)/2} \beta_l (f_{M-l} - f_{M+l})
        + h \sum_{1 \le |j| \le n} \gamma_j f_j
        + O(h^n),
\]
and the scheme for $I_2$ is
\[
    I_2 = h \left[\sum_{i=1}^{M-1} f_{-i} + \frac 12 f_{-M} \right]
        + h \sum_{l=1}^{(m-1)/2} \beta_l (f_{-M+l} - f_{-M-l})
        + h \sum_{1 \le |j| \le n} \gamma_j f_{-j}
        + O(h^n).
\]
By periodicity, we may identify $f_i$ with $f_{i+2M}$ for all $i$.
It follows that an $n$th order quadrature for $I_1 + I_2$ is
\begin{align*}
    I_1 + I_2 &= h \left[ \sum_{1 \le |i| \le M-1} f_i + f_M \right] 
        + h \sum_{1 \le |j| \le n} \gamma_j (f_j + f_{-j})
        + O(h^n) \\
    &= \sum_{\substack{j=-M+1 \\ j \ne 0}}^M w_j f_j + O(h^n)
\end{align*}
with
\[
    w_j = 
    \begin{cases}
        h(1+\gamma_j + \gamma_{-j}), \quad & 1 \le |j| \le n \\
        h, \quad &\text{otherwise}.
    \end{cases}
\]

We note that, by periodicity, the endpoint corrections corresponding
to the constants $\beta_l$ exactly cancel.
Moreover, we no longer need to place
additional quadrature nodes beyond $[-b,b]$ because periodicity 
identifies the new quadrature nodes with existing nodes.

Given a table of $\gamma_j$ weights, this quadrature scheme is as easy to implement as the trapezoidal rule, and yields high order convergence, with a quadrature error that is $O(h^n)$. The necessary $\gamma_j$ weights can be found in \citet[table 6]{kapur1997high} for $n=2,6,10$. Figure \ref{fig:example_KR_nodes_weights_periodic} can be compared with figure \ref{fig:example_KR_nodes_weights} to view the simplifications for periodic integrands.

\begin{figure}
    \centering
    \includegraphics[width=1\textwidth]{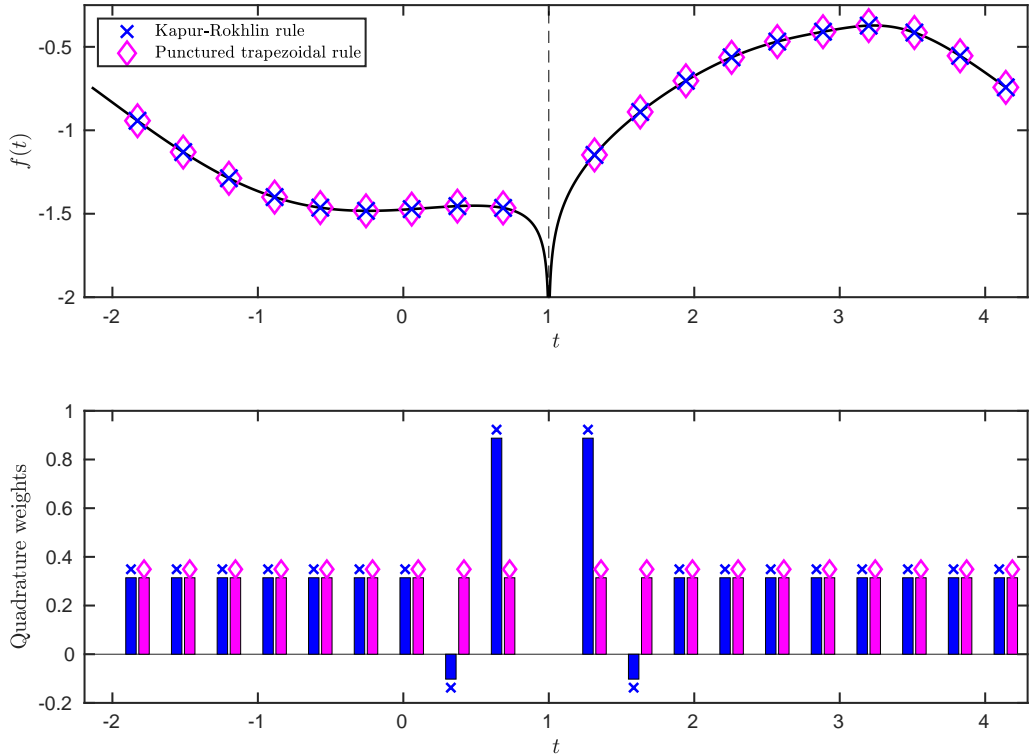}
    \caption{Nodes and weights for the Kapur-Rokhlin and punctured trapezoidal quadrature rules for estimating $\int_{t_0-\pi}^{t_0+\pi} f(t) \d t$. The function $f(t)$ is the $2\pi$-periodic integrand in equation \eqref{eq:doubleID_1D} for one of our later numerical tests with a logarithmic singularity at $t_0=1$. We have used $M=10$ and $n=2$.}
    \label{fig:example_KR_nodes_weights_periodic}
\end{figure}


\section{The virtual casing principle for toroidally axisymmetric domains}\label{sec:virtual_casing}

\subsection{Formulation of the virtual casing principle}

For axisymmetric confinement devices, the virtual casing principle is most often used to compute the poloidal flux or the poloidal magnetic field due to the toroidal current flowing in the plasma \citep{Shafranov_1972,Zakharov_1973,ZakharovPletzer,Hirshman1986}. A poloidal magnetic field ${\vec B}^{pol}$ at any point $\vec y(\phi, t) \in \Gamma$ can be expressed in terms of its poloidal flux function $\psi(r,z)$ and the parameterization $(\phi,t) \mapsto \vec y(\phi, t)$ by \citep{freidberg2014ideal}
\begin{align}
    \vec B^{pol}(\vec y(\phi,t)) 
    &= \grad \psi(r(t),z(t)) \cross \grad \phi \nonumber \\
    &= \grad \psi(r(t),z(t)) \cross \left(\frac {\vec e_\phi(\phi)}{r(t)} \right) \nonumber \\
    &= -\frac 1{r(t)} \pder\psi z(r(t), z(t)) \vec e_r(\phi) 
        + \frac 1{r(t)} \pder\psi r(r(t), z(t)) \vec e_z.
        \label{eq:B_axisym}
\end{align}

Consider an axisymmetric plasma confined by external coils in equilibrium. The poloidal field $\vec B^{pol}$ at any location is the sum of the poloidal field $\vec B_{ext}^{pol}$ due to the external coils, and of the poloidal field $\vec B_V^{pol}$ due to the plasma current. The field $\vec B_V^{pol}$ is given for all $\vec x \in \R^3$ by the Biot-Savart law:
\begin{equation}
    \vec B_V^{pol}(\vec x)=\frac{\mu_{0}}{4\pi} \iiint_{\Omega}J^{tor}_{V}(\vec y) \vec e_\phi(\vec y) \cross \frac{\vec x - \vec y}{\norm{\vec x - \vec y}^3} \d \vec y
    \label{eq:volume_int}    
\end{equation}
where $\mu_{0}$ is the permeability of free space, and $J^{tor}_{V}$ is the toroidal current density in the plasma. Equation \eqref{eq:volume_int} is a volume integral, which is expensive to evaluate numerically. The virtual casing principle gives a formula for $\vec B^{pol}_V$ that depends only on the full field $\vec B^{pol}$ at the plasma boundary, and only requires the evaluation of a surface integral (i.e. line integral for axisymmetric domains) \citep{Shafranov_1972,Zakharov_1973,hanson2015virtual}. 
In experimental settings, such a representation is useful since only the total magnetic field $\vec B^{pol}$ may be directly measurable \citep{hutchinsonprinciples}. In theoretical settings, the total magnetic field may be computed by solving the Grad-Shafranov equation \citep{shafranov1958magnetohydrodynamical,grad1958proceedings} numerically with a fixed-boundary solver.
Specifically, the virtual casing principle states that if $\Gamma$ is the flux surface bounding the plasma, then $\vec B_V^{pol}$ can be written in terms of a field generated by the toroidal surface current $\vec J^{tor}_S$ such that $\mu_0 \vec J^{tor}_S = -\vec n \cross \vec B^{pol}$, according to \citet{hanson2015virtual}:
\begin{equation} \label{eq:virtual_casing}
    \vec B^{pol}_V(\vec x) = 
    \frac 1{4\pi} \iint_{\Gamma} \left[
    \frac{(\vec n(\vec y) \cross \vec B^{pol}(\vec y)) \cross (\vec x - \vec y)}{\norm{\vec x - \vec y}^3}
    \right] \d \Gamma(\vec y) 
    + 
    \begin{cases}
        \vec B^{pol}(\vec x) \quad &\vec x \in \Omega \\
        \vec B^{pol}(\vec x) / 2 \quad &\vec x \in \Gamma \\
        \vec 0 \quad &\vec x \notin \overline{\Omega}.
    \end{cases}
\end{equation}

For certain applications, one is only interested in the normal component of the poloidal magnetic field \citep{MerkelJCP,HirshmanNEMEC,Merkel_NESCOIL,Landreman_REGCOIL,Zhu_FOCUS}. The previous equation then leads to the more compact form
\begin{equation} \label{eq:virtual_casing_normal}
    \vec n(\vec x) \vec\cdot \vec B^{pol}_V(\vec x) = 
    \frac 1{4\pi} \vec n(\vec x) \vec\cdot \iint_{\Gamma} \left[
    \frac{(\vec n(\vec y) \cross \vec B^{pol}(\vec y)) \cross (\vec x - \vec y)}{\norm{\vec x - \vec y}^3}
    \right] \d \Gamma(\vec y)
\end{equation}
for $\vec x \in \Gamma$. 

The reduction of the integral necessary to compute the field or its normal component from a volume integral to a surface integral is convenient from the point of view of the limited number of values that need to be specified as inputs, and also from the point of view of the computational cost of the integration \citep{Lazerson_2013}. The surface integral in \eqref{eq:virtual_casing} is significantly faster to evaluate than the volume integral \eqref{eq:volume_int}, although certain codes still choose to compute the latter \citep{Hanson_2009,Marx_2017}. 

For axisymmetric situations, one may further take advantage of the axisymmetry of $\vec B^{pol}$ to integrate with respect to $\phi$ analytically, and reduce \eqref{eq:virtual_casing} and \eqref{eq:virtual_casing_normal} to one-dimensional integrals, which are even less computationally expensive. However, one encounters a mathematical and computational difficulty if one does so, because the surface integrals in \eqref{eq:virtual_casing} and \eqref{eq:virtual_casing_normal} are in fact improper integrals, which must be understood in the Cauchy principal value sense. This is what we discuss in the following section.



\subsection{Numerical evaluation of the normal component of the virtual-casing magnetic field in axisymmetric systems}

\subsubsection{Circumventing integrals in the principal value sense}
Most applications in magnetic confinement fusion rely on version \eqref{eq:virtual_casing_normal} of the virtual casing principle, in which one wants to compute the normal component of the poloidal magnetic field $\vec B_V^{pol}$ at a boundary point $\vec x\in\Gamma$. One could in principle calculate this normal component by first computing {\em all} the components of $\vec B_V^{pol}$ by the virtual casing principle \eqref{eq:virtual_casing}, and then computing $\vec n \vec\cdot \vec B_V^{pol}$ by a straightforward inner product. In other words, one first evaluates the double integral in \eqref{eq:virtual_casing_normal}, and then takes the dot product with the normal vector $\vec n$ to the surface $\Gamma$ at the point $\vec x$ of interest. This method has the advantage that it produces all components of the poloidal magnetic field $\vec B_V^{pol}$ as intermediary results. However, it also has the disadvantage that one must use a careful principal value integration procedure to interpret the virtual casing principle for $\vec x \in \Gamma$, as discussed in theorem \ref{th:need_pv} and its proof in the appendix. The high order singularity cancellation quadrature scheme we recently proposed for singular integrals on general non-axisymmetric surfaces \citep{malhotra2019efficient} automatically yields the appropriate principal value of the integral. However, it does so thanks to the intrinsic two-dimensional nature of the integral. It does not reduce to a simple and efficient one-dimensional quadrature scheme for the principal value of the integral, as is needed for axisymmetric applications. Similarly, we are not aware of a version of the Kapur-Rokhlin quadrature scheme designed to calculate the Cauchy principal value of the virtual-casing integral.

To address this difficulty, we consider an alternative method to compute the normal poloidal field $\vec n \vec\cdot \vec B_V^{pol}$, based on calculating the vector potential $\vec A_S$ produced by the surface current $\vec J^{tor}_S$, and then obtaining $\vec n \vec\cdot \vec B_V^{pol}$ as the tangential derivative of the poloidal flux $\psi_S$ caused by the surface currents, which is easily expressed in terms of $\vec A_S$. This method is also used, in a slightly different form, by the stellarator optimization code ROSE \citep{Drevlak_2018}, but ROSE does not combine it with a high order quadrature scheme. The vector potential is defined by
\begin{equation}
    \vec A_S(\vec x) = \frac {\mu_0}{4 \pi} 
        \iint_\Gamma \frac{\vec J^{tor}_S(\vec y)}{\norm{\vec x - \vec y}}
        \d \Gamma(\vec y).
    \label{eq:vecpot_original}
\end{equation}
and we recognize this expression as a vector of component-wise single layer potentials for Laplace’s equation in three dimensions, as introduced in Section \ref{sec:layer_pot}.

In contrast to the first approach, our method does not produce all components of $\vec B^{pol}_V$. On the other hand, it only requires weakly singular integrals, since it only requires evaluations of single-layer potentials in \eqref{eq:vecpot_original}, as we will prove in Section \ref{sec:analytic}. In the next section, we describe in detail this alternative method for computing the normal magnetic field.

\subsubsection{Normal component of the magnetic field as a tangential derivative}

As in the derivation of the virtual-casing principle \citep{Shafranov_1972,Zakharov_1973,hanson2015virtual}, one can interpret the plasma boundary $\Gamma$ of the axisymmetric plasma as a perfectly conducting shell which contributes to confining the plasma via the poloidal magnetic field $\vec B^{pol}_S$ generated by the toroidal surface current density $\vec J^{tor}_S$ flowing in $\Gamma$. 
At equilibrium, the poloidal surface current field $\vec B^{pol}_S$ matches the poloidal field $\vec B_{ext}^{pol}$ from the external confining coils.
The decomposition $\vec B^{pol} = \vec B^{pol}_V + \vec B^{pol}_S$ then immediately yields $\vec n \vec\cdot \vec B^{pol}_V = -\vec n \vec\cdot \vec B^{pol}_S$.
Now, recall that we may represent $\vec B^{pol}_S$ in axisymmetry via \eqref{eq:B_axisym} to express its normal component at $(r, \phi, z)$ as
\[
    \vec n \vec\cdot \vec B^{pol}_S
    = \vec n \vec\cdot (\grad \psi_S \cross \grad \phi)
    = -\frac 1r (\vec n \cross \vec e_\phi(\phi)) \vec\cdot \grad \psi_S
\]
by the circular shift identity of the vector triple product. Let $\vec x$ correspond to the parameter pair $(\phi_0, t_0)$ with $\vec x = (R, \phi_0, Z) = (r(t_0), \phi_0, z(t_0))$. We define the tangent vector $\vec t$ at a point $\vec x \in \Gamma$ as 
\begin{align*}
    \vec t(\vec x) 
    &= \vec n(\vec x) \cross \vec e_\phi(\phi_0) \\
    &= \frac{(z'(t_0) \vec e_r(\phi_0) - r'(t_0) \vec e_z ) \cross \vec e_\phi(\phi_0)}
        {\sqrt{r'(t_0)^2 + z'(t_0)^2}} \\
    &= \frac{ r'(t_0) \vec e_r(\phi_0) + z'(t_0) \vec e_z }
        {\sqrt{r'(t_0)^2 + z'(t_0)^2}}.
\end{align*}
It follows that
\begin{align*}
    \vec n(\vec x) \vec\cdot \vec B^{pol}_V(\vec x)
    &= \frac 1{r(t_0)} \vec t(\vec x) \vec\cdot \grad \psi_S(r(t_0), z(t_0)) \\
    &= \frac 1{r(t_0)\sqrt{r'(t_0)^2 + z'(t_0)^2}}
        \left( \pder{\psi_S} r r'(t_0) + \pder{\psi_S} z z'(t_0) \right) \\
    &= \frac 1{\mathcal J(t_0)}
        \frac{\d \psi_S(r(t_0), z(t_0))}{\d t_0}.
\end{align*}
where $\mathcal J(t_0) = r(t_0)\sqrt{r'(t_0)^2 + z'(t_0)^2}$ is the Jacobian of the parameterization introduced in Section \ref{sec:torsurf_desc}. 

Finally, the poloidal flux function and the vector potential are related at $\vec x$ by the relation \citep{jardin2010computational,freidberg2014ideal}
\begin{equation} \label{eq:flux_def}
    \psi_S(R, Z) = R \vec e_\phi(\phi_0) \vec\cdot \vec A_S(\vec x).
\end{equation}
Given only point evaluations of $\vec A_S$ at equispaced parameters $\{t_i\}$, one can compute $\d \psi_S / \d t$ with high order accuracy at each $t_i$ by Fourier differentiation. This leads to our high order accurate approach to the virtual casing principle in axisymmetric geometries. One computes $\vec n \vec\cdot \vec B^{pol}_V$ on $\Gamma$ by evaluating a weakly singular integral with a high order accurate quadrature rule, and then Fourier differentiating. 
We show in section \ref{sec:KR_vecpot} that Kapur-Rokhlin quadrature provides a simple way to obtain high order accuracy for the weakly singular integral. Before we do so, we prove in the next section that it indeed is a weakly singular integral, whose properties satisfy the requirements to obtain high accuracy with the Kapur-Rokhlin quadrature rule.


\section{Analytical Results for Singular Integrals}\label{sec:analytic}

\subsection{Vector potential singularity} \label{sec:ana_vecpot}
 By the virtual casing principle, we have $\mu_0 \vec J^{tor}_S = - \vec n \cross \vec B^{pol}$, so we may equivalently write the vector potential as
\begin{equation}
\label{eq:vecpot_int}
    \vec A_S(\vec x) = -\frac 1{4\pi} \iint_\Gamma 
    \frac{\vec n(\vec y) \cross \vec B^{pol}(\vec y)}
        {\norm{\vec x - \vec y}} 
    \d \Gamma(\vec y),
\end{equation}
which satisfies $\vec B^{pol}_S = \vec \nabla \cross \vec A_S$.  As mentioned earlier, we may view this expression as a vector of component-wise single layer potentials for
Laplace’s equation in three dimensions. Physically, it is clear that the only non-zero component of the vector potential is in the $\vec e_\phi$ direction. Furthermore, \citet[][\S V]{chance1997vacuum} shows that after integrating in the toroidal angle $\phi$, the one-dimensional integrands of both the single layer potential $\mathcal S[\mu_0\vec J^{tor}_S]$ and the double layer potential $\mathcal D[\mu_0\vec J^{tor}_S]$ have integrable, logarithmic singularities. Chance shows this in a modified coordinate system
using the poloidal flux function 
as a coordinate. Next,
we will verify these results for the virtual casing principle 
in standard cylindrical coordinates.

\subsection{Analytic Reduction to a Line Integral} \label{sec:ana_mainthm}
We shall analytically simplify the integral in \eqref{eq:vecpot_int} by integrating over the toroidal angle. When we do so, we shall introduce the complete elliptic integrals of the first and second kind, which are defined as
\[
K(k^2) = \int_0^{\pi/2} \frac{\d \theta}{\sqrt{1 - k^2 \sin^2\theta}}
\quad \text{and} \quad
E(k^2) = \int_0^{\pi/2} \sqrt{1 - k^2 \sin^2\theta} \d \theta,
\]
respectively.
Our main analytical result gives a formula for the vector potential of the virtual casing
principle 
as a one-dimensional integral.

\begin{theorem} \label{th:vecpot}
Let $\Gamma$ be a surface of revolution with a generating curve $\gamma$ that satisfies the assumptions of section \ref{sec:torsurf_desc}.
Consider the vector potential
\begin{displaymath}
    \vec A_S(\vec x) = -\frac 1{4\pi} \iint_\Gamma 
    \frac{\vec n(\vec y) \cross \vec B^{pol}(\vec y)}
        {\norm{\vec x - \vec y}} 
    \d \Gamma(\vec y)
\end{displaymath}
for $\vec x = (R,\phi_0,Z) = (r(t_0), \phi_0, z(t_0))$ in cylindrical coordinates.
Define the quantities
\[
    \alpha = \alpha(t; \vec x) = r(t)^2 + R^2 + (Z - z(t))^2
    \quad \text{and} \quad
    \beta = \beta(t; \vec x) = 2 R r(t)
\]
and set the modulus
\[
    k^2 = k(t; \vec x)^2 = \frac{2\beta}{\alpha + \beta}
    = \frac{4 R r(t)}{(R + r(t))^2 + (Z - z(t))^2}.
\]
Then the vector potential can be expressed as
\begin{equation} \label{eq:vecpot_1d}
    \vec A_S(\vec x) = -\vec e_\phi(\phi_0) \int_0^L \mathcal A(t; \vec x) \d t
\end{equation}
with the scalar-valued integrand
\[
    \mathcal A(t; \vec x)
    = \frac 1{4\pi} \left( \frac 4{\sqrt{\alpha + \beta}} \right)
    \left[\pder\psi z r'(t) - \pder\psi r z'(t) \right]
    \left( \frac 2{k^2} \Big(K(k^2) - E(k^2) \Big) - K(k^2) \right).
\]
\end{theorem}
\begin{proof}
The unit outward normal vector of $\Gamma$ at $\vec y$ 
is given (by our assumption on the orientation of $\gamma$) by
\begin{align*}
    \vec n(\vec y(\phi,t)) 
    &= \pder{\vec y}\phi\cross\pder{\vec y}t \Bigg/
        \norm{\pder{\vec y}\phi\cross\pder{\vec y}t} \\
    &= \frac{z'(t) \vec e_r(\phi) - r'(t) \vec e_z}
        {\sqrt{r'(t)^2 + z'(t)^2}}.
\end{align*}
From the representation \eqref{eq:B_axisym} of $\vec B^{pol}$ in axisymmetry,
it follows that
\begin{equation} \label{eq:ncrossB_param}
    (\vec n(\vec y(\phi,t)) \cross \vec B^{pol}(\vec y(\phi,t))) \mathcal J(t) 
    = \left[\pder\psi z r'(t) - \pder\psi r z'(t) \right]
    \vec e_\phi(\phi).
\end{equation}

Next, we consider the difference
\begin{align*}
    \vec x - \vec y(\phi, t)
    &= [R\vec e_r(\phi_0) + Z\vec e_z] - [r(t)\vec e_r(\phi)
        +  z(t)\vec e_z] \\
    &= ( R\cos\phi_0 - r(t)\cos\phi ) \vec e_x 
        + ( R\sin\phi_0 - r(t)\sin\phi ) \vec e_y 
        + (Z- z(t))\vec e_z,
\end{align*}
which we have expressed in both cylindrical and rectangular coordinates.
Recall that the unit vector $\vec e_z$ is identical in both coordinate systems, and the remaining rectangular unit vectors $\vec e_x$ and $\vec e_y$ are related to standard cylindrical unit vectors by
\begin{align*}
    \vec e_r(\phi) &= \cos\phi \, \vec e_x + \sin\phi \, \vec e_y 
    \quad \text{and} \\
    \vec e_\phi(\phi) &= -\sin\phi \, \vec e_x + \cos\phi \, \vec e_y.
\end{align*}
From the representation
in rectangular coordinates, we use the trigonometric identity
$\cos(\phi_0-\phi)=\cos\phi_0\cos\phi + \sin\phi_0\sin\phi$
and immediately obtain
\begin{equation} \label{eq:diffnorm}
    \norm{\vec x - \vec y(\theta, t)}^2
    = R^2 + r(t)^2 + (Z - z(t))^2 - 2 R r(t) \cos(\phi_0-\phi)
    = \alpha - \beta \cos(\phi_0 - \phi).
\end{equation}
We have now shown that the surface integral \eqref{eq:vecpot_int}
for the vector potential is equivalent to 
the following double integral over a rectangle in the parameter domain:
\[
    \vec A_S(\vec x) = -\frac 1{4\pi}
    \int_0^L \left[\pder\psi z r'(t) - \pder\psi r z'(t) \right]
    \int_0^{2\pi}
    \frac{\vec e_\phi(\phi)}{\sqrt{\alpha(t) - \beta(t) \cos(\phi_0 - \phi)}}
    \d \phi \d t.
\]
We may analytically evaluate the inner integral using trigonometric identities and known integral formulae.
We use the identities 
\[
    \begin{cases}
        \sin(\phi+\phi_0) = \sin\phi\cos\phi_0 + \cos\phi\sin\phi_0 
        \\
        \cos(\phi+\phi_0) = \cos\phi\cos\phi_0 - \sin\phi\sin\phi_0
    \end{cases}
\]
to compute that
\begin{align*}
    \int_0^{2\pi} \frac{-\sin\phi \, \vec e_x + \cos\phi \, \vec e_y}
        {\sqrt{\alpha - \beta\cos(\phi_0 - \phi)}} \d \phi
    &= \int_{0}^{2\pi} 
        \frac{-\sin(\phi+\phi_0) \, \vec e_x + \cos(\phi+\phi_0) \, \vec e_y}
            {\sqrt{\alpha - \beta\cos\phi}} \d \phi \\
    &= \vec e_\phi(\phi_0) \int_0^{2\pi} 
        \frac{\cos\phi \, \d \phi}{\sqrt{\alpha - \beta\cos\phi}} 
        - \vec e_r(\phi_0) \int_0^{2\pi} 
        \frac{\sin\phi \, \d \phi}{\sqrt{\alpha - \beta\cos\phi}} \\
    &= \vec e_\phi(\phi_0) \int_0^{2\pi} 
        \frac{\cos\phi \, \d \phi}{\sqrt{\alpha - \beta\cos\phi}}.
\end{align*}
Here, we have used the fact that
\[
    \int_0^{2\pi} \frac{\sin\phi \, \d \phi}{\sqrt{\alpha - \beta \cos\phi}} = 0,
\]
which holds because it is the integral of an odd function over a single period.

Now, the remaining integral can be expressed in terms of the complete elliptic
integrals, through the identity
\begin{align}
    \int_0^{2\pi} \frac{\cos\phi \, \d \phi}{\sqrt{\alpha - \beta\cos\phi}}
    &= 2\int_{-\pi/2}^{\pi/2} \frac{\cos (2\phi+\pi) \, \d \phi}{\sqrt{\alpha - \beta\cos (2\phi+\pi)}} 
        \nonumber
        \\
    &= 2\int_{-\pi/2}^{\pi/2} \frac{-(1 - 2\sin^2\phi) \, \d \phi}{\sqrt{\alpha + \beta(1 - 2\sin^2\phi)}} 
        \nonumber
        \\
    &= \frac{-4}{\sqrt{\alpha+\beta}} \int_0^{\pi/2} \frac{1-2\sin^2\phi}
        {\sqrt{1 - k^2 \sin^2\phi}} \d \phi 
        \nonumber
        \\
    &= \frac{-4}{\sqrt{\alpha+\beta}} \left( K(k^2) - \frac 2{k^2} (K(k^2) - E(k^2))\right). \label{eq:gr_transform}
\end{align}
We obtain the last equality \eqref{eq:gr_transform} from the definition of $K$ and from a formula of
\citet[2.584-4]{gradshteyn2014table}.
The desired result now immediately follows.

\end{proof}

The Kapur-Rokhlin quadrature rule that we analyzed in section \ref{sec:KR_periodic} applies to integrands with a logarithmic singularity. With the result we just obtained, we can verify, by analyzing the behavior of the complete elliptic integrals, that the integrand $\mathcal A$ is indeed logarithmically singular as $t \to t_0$, in agreement with the results from \citet[][\S V]{chance1997vacuum}. Specifically, as the modulus $k$ tends to $1$, the second-kind integral $E$ is continuous and bounded, and the first-kind elliptic integral $K$ is logarithmically singular. The mapping $t \mapsto k(t; \vec x)^2$ is continuous, so $E(k(t;\vec x)^2)$ is continuous and $K(k(t;\vec x)^2)$ is logarithmically singular as $t \to t_0$.
Readers interested in more detail regarding these results are referred to lemma \ref{lem:K_asymptotics} in the appendix.

\section{A Kapur-Rokhlin scheme for the virtual casing principle} \label{sec:KR_vecpot}

We may compute the vector potential by equation \eqref{eq:vecpot_1d} of theorem \ref{th:vecpot} using the Kapur-Rokhlin quadrature rule for periodic functions 
introduced in section \ref{sec:KR_periodic}.
Let $\vec x \in \Gamma$ be given, corresponding to parameters $(\phi_0,t_0)$. We reiterate that a univariate integral expression for the vector potential is
\[
    \vec A_S(\vec x) = -\vec e_\phi(\phi_0) \int_0^L \mathcal A(t; \vec x) \d t.
\]
The Kapur-Rokhlin quadrature of section \ref{sec:KR_periodic} applies directly because
the integration interval $[0,L]$ is identical, by periodicity, to the symmetric interval $[t_0-L/2, t_0+L/2]$ and 
because the integrand $\mathcal A(t; \vec x)$ is logarithmically singular as $t \to t_0$.

Given $M \in \mathbb N$, we generate $2M-1$ quadrature nodes $t_i = t_0 + ih$ for $i = \pm1, \dots, \pm(M-1),M$ with spacing $h=L/(2M)$. We evaluate $\mathcal A_i = \mathcal A(t_i; \vec x)$, and it follows that
\begin{equation} \label{eq:KR_vecpot}
    \int_0^L \mathcal A(t; \vec x) \d t
    = h \left[ \sum_{1 \le |i| \le M-1} \mathcal A_i + \mathcal A_M \right] 
        + h \sum_{1 \le |j| \le n} \gamma_j (\mathcal A_j + \mathcal A_{-j})
        + O(h^n).
\end{equation}

Since periodicity has removed our considerations about expanding the integration domain, the parameter $n$ can be taken as large as we like (as long as $M \ge n$ to define the quadrature rule). However, in practice this Kapur-Rokhlin scheme is known to be unstable for $n$ larger than about 10 because the weights $\gamma_j$ are sign-indefinite and grow large in magnitude.

\section{Numerical results}\label{sec:numerics}
In this section, we illustrate the Kapur-Rokhlin quadrature schemes from sections \ref{sec:KR_periodic} and \ref{sec:KR_vecpot} in two different calculations. Throughout, we will test the schemes for an axisymmetric plasma boundary given by the level set $\psi=0$ of the poloidal flux function given by
\[
    \psi(r,z) = \frac{\kappa F_B}{2 R_0^3 q_0}
    \left[\frac 14 (r^2 - R_0^2)^2 + \frac 1{\kappa^2} r^2 z^2 - a^2 R_0^2 \right].
\]
This flux function is a solution to the Grad-Shafranov equation with the Solov'ev profiles $\mu_{0}p(\psi)=-[F_{B}(\kappa + 1/\kappa)/( R_0^{3}q_0)]\psi$ and $F(\psi)=F_B$, where $p(\psi)$ is the plasma pressure profile, and $F(\psi)=rB_{\phi}$, with $B_\phi$ the toroidal magnetic field \citep{lutjens1996chease,lee2015ecom}. The parameters $R_0$ and $q_0$ may be interpreted as the major radius and safety factor at the magnetic axis, and $\kappa$ and $a$ as the elongation and minor radius of the plasma boundary. All numerical tests in this manuscript use the fusion relevant values $F_B = R_0 = q_0 = 1$ and $\kappa = 1.7$ and $a = 1/3$. The level set $\psi = 0$ may be parameterized by the functions \citep{lutjens1996chease,lee2015ecom}
\begin{align*}
    (r(t))^2 &= R_0^2 + 2aR_0 \cos t 
    \quad \text{and} \\
    z(t) &= \kappa a \frac{R_0}{r(t)}\sin t
\end{align*}
for $t \in [0, 2\pi]$.

\subsection{Double layer identity}

For our first numerical verification, we consider an identity associated with harmonic functions.
Consider the Green's function $G(\vec x, \vec y) = (4 \pi \norm{\vec x - \vec y})^{-1}$
for Laplace's equation in three dimensions. It satisfies the double layer jump condition \citep{malhotra2019efficient}
\[
    \iint_\Gamma \pder{G(\vec x, \vec y)}{\vec n(\vec y)} \d \Gamma(\vec y)
    = \frac 1{4\pi} \iint_\Gamma 
        \frac{\vec n(\vec y) \cdot (\vec x - \vec y)} {\norm{\vec x - \vec y}^3}
    \d \Gamma(\vec y)
    = - \frac 12
\]
for $\vec x \in \Gamma$. It follows that
\[
    1 + \frac 1{2\pi} \iint_\Gamma 
    \frac{\vec n(\vec y) \cdot (\vec x - \vec y)} {\norm{\vec x - \vec y}^3}
    \d \Gamma(\vec y) = 0,
\]
again for $\vec x \in \Gamma$.
Following identical methodology to the proof of theorem \ref{th:vecpot}, we integrate out the toroidal angle analytically to obtain the one-dimensional integral identity
\begin{equation}
    \label{eq:doubleID_1D}
    1 + \frac 1{2\pi} \int_0^L \frac{4r}{(\alpha+\beta)^{3/2}} 
    \left\{ -\frac{2 z' R}{k^2} K(k^2) + \left(
        \frac{2 z' R}{k^2} + \frac{z'(R-r) - r'(Z-z)}{1-k^2}
        \right) E(k^2)
    \right\} \d t
    = 0.
\end{equation}
In the above expression, we have suppressed the dependence of $\{r,z,r',z',\alpha,\beta,k^2\}$ on $t$ for clarity. As usual, we have also used the identification $\vec x = (R, \phi_0, Z) = (r(t_0), \phi_0, z(t_0))$.
The integrand is logarithmically singular because of the presence of the singular elliptic integral $K(k^2)$, and because the other seemingly singular coefficient is actually bounded
when $r''(t_0)$ and $z''(t_0)$ exist (which is the case in our example), with
\[
    \lim_{t \to t_0} \frac{z'(t)(R-r(t)) - r'(t)(Z-z(t))}{1-k(t)^2}
    = 2 R^2 \left( \frac{r''(t_0) z'(t_0) - r'(t_0) z''(t_0)}{r'(t_0)^2 + z'(t_0)^2} \right).
\]
We conclude that the integral in \eqref{eq:doubleID_1D} is of the Kapur-Rokhlin form from section \ref{sec:KR_periodic}. 

Figure \ref{fig:double_layer} illustrates the performance of the 10th order periodic Kapur-Rokhlin quadrature scheme for verifying the identity \eqref{eq:doubleID_1D}. We compare this method with the alternating trapezoidal rule, which uses common quadrature weight $h=\pi/M$ and quadrature nodes $\tilde t_i = t_0 + (i - \frac 12) h$ for $i=0, \pm1, \dots, \pm(M-1), M$ that straddle the singularity at $t_0$.
We find that the Kapur-Rokhlin scheme achieves the theoretical 10th order accuracy, and that we obtain the full accuracy of the quadrature scheme using about 175 quadrature nodes. Because the behavior of the integrand is not generally symmetric about $t=t_0$, the alternating trapezoidal rule performs poorly. In order to avoid unfairly representing the performance of singularity subtraction schemes currently used in plasma physics applications, we have not attempted to code our own versions of them to compare their accuracy and their convergence rate in this figure. However, we explain in \citet{malhotra2019efficient} why all these singularity subtraction schemes are expected have low order convergence, with second order convergence expected for most of them.

\begin{figure}
    \centering
    \includegraphics[width=0.8\textwidth]{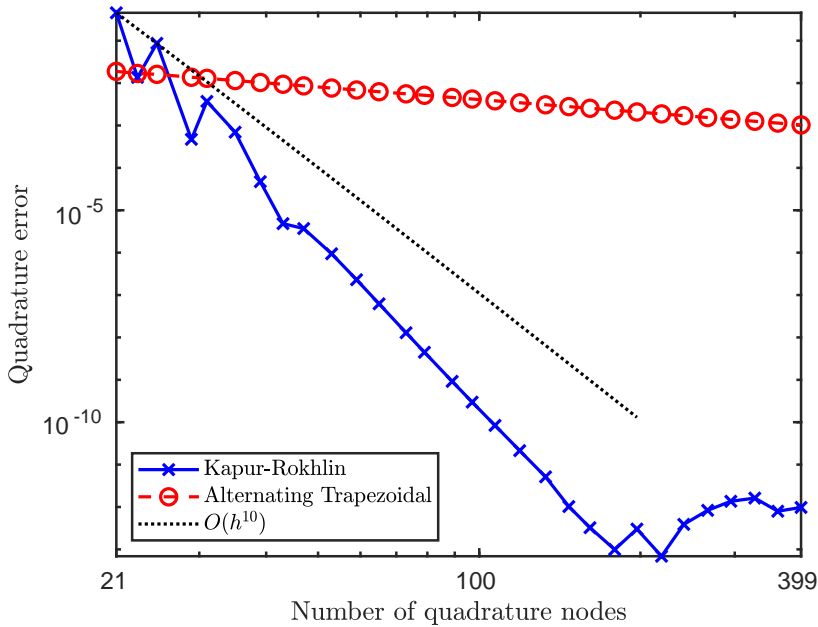}
    \caption{Error convergence for the double layer identity \eqref{eq:doubleID_1D} at $t_0 = 1$.}
    \label{fig:double_layer}
\end{figure}

\subsection{Virtual casing principle}
For our second test, we evaluate the accuracy of our method for calculating the normal component of the magnetic field due to the plasma current on the plasma boundary, i.e. $\vec n \vec\cdot \vec B^{pol}_V$ on $\Gamma$. The plasma equilibrium we consider is the Solov'ev equilibrium described above \citep{lutjens1996chease,lee2015ecom}. Since we do not know the analytic solution to this problem, we compare our implementation with an existing high order accurate implementation from \citet{malhotra2019efficient}. The method used in that implementation is different from the approach presented here in several ways, making it appropriate for our verification. Specifically, the code presented in \citep{malhotra2019efficient} views the plasma equilibrium as a fully three-dimensional equilibrium, and does not assume axisymmetry. Furthermore, \citet{malhotra2019efficient} obtain $\vec n \vec\cdot \vec B^{pol}_V$ on $\Gamma$ directly via a direct evaluation of \eqref{eq:virtual_casing}, as opposed to first computing the vector potential. Finally, the Cauchy principal value of \eqref{eq:virtual_casing} is numerically evaluated via a partition of unity scheme to handle the singulary of the integrand. We take the result from a high resolution calculation of \citet{malhotra2019efficient} for this problem as the ground truth, against which we test the accuracy of our approach as a function of the number of quadrature points.

Figure \ref{fig:ndotB_code} illustrates our results when using the 10th order Kapur-Rokhlin scheme for this problem. Again, we find that the Kapur-Rokhlin scheme achieves the theoretical convergence rate, and that the accuracy has converged by about 400 quadrature nodes. In both this example and the double layer identity example, we observe that the Kapur-Rokhlin scheme errors do not converge all the way to machine precision. This is one known drawback of the Kapur-Rokhlin methods, driven partially by the instabilities caused by the correction weights. Nonetheless, in contexts where full machine precision is not necessary, this scheme provides high accuracy and is easy to implement by reading the correction weights $\gamma_j$ from a table of pre-computed values.

\begin{figure}
    \centering
    \includegraphics[width=0.8\textwidth]{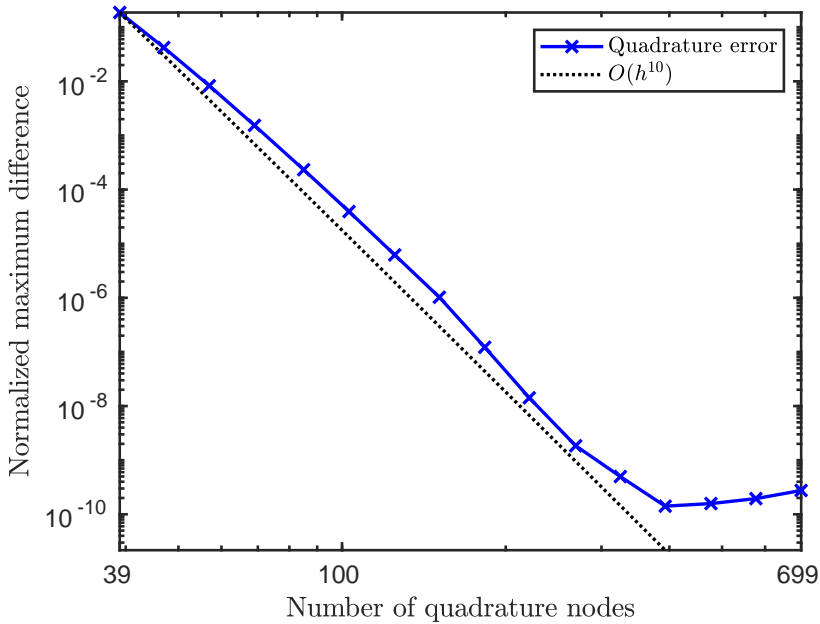}
    \caption{Comparison with existing code to compute the normal component of the magnetic field.}
    \label{fig:ndotB_code}
\end{figure}

\section{Conclusion}\label{sec:conclusion}
For axisymmetric confinement fusion systems, a direct implementation of the Kapur-Rokhlin quadrature scheme yields high accuracy for the evaluation of the layer potentials commonly encountered in magnetostatic and magnetohydrodynamic calculations. Using the table of quadradure weights given in the original article by Kapur and Rokhlin \citep{kapur1997high}, the scheme is as easy to implement as the trapezoidal rule, and, unlike commonly used methods, does not require any manipulation of the singular integrands. We demonstrated how to implement it for the evaluation of a double-layer potential and for the virtual casing principle, obtaining 10th order convergence in both cases. 

At the moment, our approach is restricted to singular integrals along smooth boundaries, and cannot be applied to magnetic surfaces with one or several X-points. The development of an efficient and accurate quadrature scheme for this importance case is the subject of ongoing work, with results to be reported at a later date.

\section*{Acknowledgements}
The authors would like to thank Leslie Greengard and Mike O'Neil for insightful discussions. Evan Toler was supported by the National Science Foundation Graduate Research Fellowship under grant no. 1839302. A. J. Cerfon was supported by the United States Department of Energy, Office of Fusion Energy Sciences, under grant no. DE-FG02-86ER53223

\section*{Declaration of Interests}
The authors report no conflict of interest.

\bibliographystyle{jpp}
\bibliography{biblio}

\appendix

\section{}\label{appA}
In this appendix, we prove that the virtual casing principle \eqref{eq:virtual_casing} is undefined by standard integration when the evaluation point lies on the boundary.
In order to prove this, we require the following two lemmas
concerning the asymptotic behavior of the complete elliptic
integral $K(k(t)^2)$ for values $t$ close to $t_0$. Throughout, we will identify the evaluation point $\vec x \in \Gamma$ with the parameters $(\phi_0, t_0)$, so that $\vec x = (r(t_0), \phi_0, z(t_0)) = (R, \phi_0, Z)$.

\begin{lemma} \label{lem:k2_expansion}
Assume $r(t)$ and $z(t)$ are $L$-periodic $C^1[0,L]$ parameterization functions.
Define the modulus by
\[
    k(t)^2 = \frac{4 R r(t)}{(R + r(t))^2 + (Z - z(t))^2}.
\]
Let $|t-t_0| > 0$ be sufficiently small. Then 
\[
    k(t)^2 = \left(1 + 
    (t-t_0)^2 \left(\frac{r'(s_r)^2 + z'(s_z)^2}
        {4R(R+ (t-t_0) r'(s_r))}
    \right)\right)^{-1}
\]
and
\[
    1 - k(t)^2 = \left(1 + 
    (t-t_0)^{-2} \left(\frac{4R(R+ (t-t_0) r'(s_r))}
        {r'(s_r)^2 + z'(s_z)^2}
    \right)\right)^{-1}
\]
for some $s_r,s_z$ between $t_0$ and $t$.
\end{lemma}
\begin{proof}
We use the Taylor expansions
\begin{align*}
    r(t) &= r(t_0) + (t-t_0) r'(s_r) = R + (t-t_0) r'(s_r) \\
    z(t) &= z(t_0) + (t-t_0) z'(s_z) = Z + (t-t_0) z'(s_z)
\end{align*}
for some $s_r, s_z$ between $t_0$ and $t$.
Inserting this into the modulus, we obtain
\begin{align*}
    k(t)^2 
    &= \frac{4 R (R + (t-t_0) r'(s_r))}
        {4R(R + (t-t_0) r'(s_r)) 
            + (t-t_0)^2 r'(s_r)^2 + (t-t_0)^2 z'(s_z)^2} \\
    &= \left( \frac{4R(R + (t-t_0) r'(s_r)) 
            + (t-t_0)^2 ( r'(s_r)^2 + z'(s_z)^2 )}
        {4 R (R + (t-t_0) r'(s_r))} \right)^{-1}
\end{align*}
Simplifying the fraction yields the desired expression for $k(t)^2$.
It is quick to verify the expression for $1-k(t)^2$ accordingly.
\end{proof}

\begin{lemma} \label{lem:K_asymptotics}
Assuming the same conditions as lemma \ref{lem:k2_expansion},
the complete elliptic integral of the first kind obeys the asymptotic behavior
\[
    K(k(t)^2) = - \log |t-t_0| + O(1)
    \qquad \text{as }
    t \to t_0.
\]
\end{lemma}

\begin{proof}
We first observe that $k(t)^2 \to 1^-$ as $t \to t_0$.
In this regime, \citet[8.113-3]{gradshteyn2014table}
gives the asymptotic expansion
\[
    K(k^2) = \log\left( \frac 4{\sqrt{1-k^2}} \right) 
        + o(1)
    \quad \text{as } k^2 \to 1^-.
\]
From lemma \ref{lem:k2_expansion}, we also have the 
explicit representation
of the dominant term
\[
    \log\left( \frac 4{\sqrt{1-k^2}} \right) 
    = \log 4 + \frac 12 \log
    \left( 
    1 +
    |t-t_0|^{-2} 
    \left(
    \frac{4R( R + (t-t_0) r'(s_r) )}
        {r'(s_r)^2 + z'(s_z)^2}
    \right)
    \right).
\]
Finally, we use the identity 
$\log(1+u) = \log u + \log(u^{-1} + 1) = \log u + o(1)$ as $u \to \infty$.
We conclude that
\begin{align*}
    K(k(t)^2) 
    &= \log 4 + \frac 12
    \log \left(|t-t_0|^{-2} \left(
    \frac{4R( R + (t-t_0) r'(s_r) )}
    {r'(s_r)^2 + z'(s_z)^2} 
    \right)
    \right) + o(1) \\
    &= \log 4 - \log |t-t_0|
    + \frac 12 \log\left(
    \frac{4R( R + (t-t_0) r'(s_r) )}
    {r'(s_r)^2 + z'(s_z)^2} 
    \right) + o(1)
\end{align*}
as $t \to t_0$. Only the second term is singular as $t \to t_0$; the others are bounded, and this completes the proof.
\end{proof}

We now have the necessary tools to prove that the integrand
obtained by the virtual casing principle is not absolutely
integrable on the parameter domain.

\begin{theorem} \label{th:need_pv}
Let $\Gamma$ be a smooth surface of revolution, and let $\vec x \in \Gamma$. 
Assume that there exists a constant $R_\text{min} > 0$ such that 
$r(t) \ge R_\text{min}$ for all $t \in [0,L]$.
Then
\[
    \iint_\Gamma \norm{
    \frac{(\vec n(\vec y) \cross \vec B^{pol}(\vec y)) 
        \cross (\vec x - \vec y)}
        {\norm{\vec x - \vec y}^3}} 
    \d \Gamma(\vec y) = \infty.
\]
\end{theorem}
\begin{proof}
Without loss of generality, we assume that the 
coordinate system is appropriately rotated so that $\vec x$ has zero toroidal angle. That is, we assume $\vec x = (R,0,Z) = (r(t_0), 0, z(t_0))$.
In this form, we recall that $\vec e_r(0) = \vec e_x$ and $\vec e_\phi(0) = \vec e_y$.
The surface integral can be rewritten in the parameter domain as 
\[
    \iint_\Gamma \norm{
    \frac{(\vec n(\vec y) \cross \vec B^{pol}(\vec y)) 
        \cross (\vec x - \vec y)}
        {\norm{\vec x - \vec y}^3}} 
    \d \Gamma(\vec y)
    =
    \int_0^L \int_0^{2\pi} \norm{\vec F(\phi,t)} \d\phi \d t,
\]
where $\vec F$ is expressible from the parameterization representations \eqref{eq:ncrossB_param} and \eqref{eq:diffnorm} as
\begin{align*}
    \vec F(\phi,t) 
    &= \frac{
        \left[\pder\psi z r'(t) - \pder\psi r z'(t) \right] \vec e_\phi(\phi)
        \cross
        \left[(R\vec e_x + Z\vec e_z) - (r(t)\vec e_r(\phi) + z(t)\vec e_z)\right]
    }
    {(\alpha(t) - \beta(t) \cos\phi)^{3/2}} \\
    &= \frac{\left[\pder\psi z r'(t) - \pder\psi r z'(t) \right]}
        {(\alpha(t)-\beta(t)\cos\phi)^{3/2}}
    \Big( 
        (Z -  z(t))
        \vec e_r(\phi)
    + \left(
        r(t)-R\cos\phi
        \right) \vec e_z
    \Big).
\end{align*}

By Tonelli's theorem, we may evaluate the integral of 
$\norm{\vec F}$ in $\phi$ first,
and show that the remaining integral in $t$ diverges. The integrands in $\phi$ can be transformed into expressions with formulae known from \citet[2.584-37 and 2.584-42]{gradshteyn2014table}, by a process identical to how we obtained \eqref{eq:gr_transform} in the proof of theorem \ref{th:vecpot}. The result is the univariate, vector-valued function
\begin{align*}
    \vec F_1(t) = &\int_{0}^{2\pi} \vec F(\phi,t) \d \phi\\
    = &\frac {2 \left[
        \pder\psi z r'(t) - \pder\psi r z'(t)
        \right]}
        {r(t) \sqrt{\alpha + \beta}}
    \Bigg\{ \frac {Z - z(t)}{R} 
        \left( 
        -K(k^2) 
        + \frac{\alpha}{\alpha - \beta} E(k^2) 
        \right) \vec e_x \\
    &\qquad\quad + \left( 
        K(k^2) 
        + \frac{r(t)^2 - R^2 - (Z-z(t))^2}{\alpha - \beta} E(k^2) 
        \right) \vec e_z \Bigg\}.
\end{align*}
As before, we have introduced the quantities
\[
    \begin{cases}
        \alpha  = \alpha(t; \vec x) = R^2 + r(t)^2 + (Z-z(t))^2 \\
        \beta   = \beta(t; \vec x)  = 2 R r(t) \\
        k^2     = k(t; \vec x)^2    = \displaystyle \frac{2\beta}{\alpha+\beta}.
    \end{cases}
\]
By the immediate comparison 
$\int_0^{2\pi} \|\vec F(\phi,t)\|\d \phi \ge \|\vec F_1(t)\|$, it
is sufficient to prove our claim by showing that $\|\vec F_1(t)\|$ 
is not integrable.
We will show that a singularity in one of the components
of $\vec F_1(t)$ must be at least as severe as $|t-t_0|^{-1}$ as $t \to 0$,
and this will prove that $\|\vec F(\theta, t)\|$ is not integrable.

First, we consider integrating $\vec e_x \vec\cdot \vec F_1(t)$ 
with the purely formal expression
\[
    \int_0^L \frac {2(Z - z(t))
        \left[\pder\psi z r'(t) - \pder\psi r z'(t) \right]}
        {R r(t) \sqrt{\alpha + \beta}}
    \left( -K(k^2) 
    + \frac{\alpha}{\alpha - \beta} E(k^2) 
    \right) \d  t.
\]
For any geometry where $r(t) \ge R_\text{min} > 0$, and for generic functions $\psi(r,z)$, the quantity
\[
    \lim_{t\to t_0}
    \frac {2 \left[\pder\psi z r'(t) - \pder\psi r z'(t) \right]}
        {R r(t) \sqrt{\alpha + \beta}}
    =
    \frac 1{R^3}
    \left[\pder\psi z r'(t) - \pder\psi r z'(t) \right]_{t=t_0} 
\]
is finite. Moreover, it is also nonzero since for any 
valid parameterization, the derivatives 
$r'(t)$ and $z'(t)$ cannot concurrently vanish for any 
fixed $t$, including $t=t_0$.
So, the behavior of the singularity in the integrand
$\vec e_x \cdot \vec F_1(t)$ depends purely on what remains.

The first term $(Z-z(t))K(k^2)$ is clearly integrable,
since $|Z-z(t)| = O(|t-t_0|)$ and $K(k(t)^2) = -\log|t-t_0| + O(1)$ 
as $t \to t_0$ by lemma \ref{lem:K_asymptotics}.
For the second term, we observe that the limit
\[
    \lim_{t\to t_0} \alpha E(k(t)^2)
    = 2 R^2 E(1) 
    = 2 R^2
\]
is again finite and nonzero, and so we are left to question:
How severe is the singularity of
\[
    \frac{Z - z(t)}{\alpha - \beta}
    =
    \frac{Z-z(t)}{(R-r(t))^2 + (Z-z(t))^2}
\]
as $t \to t_0$? Since the parameterizations $r(t)$ and 
$z(t)$ are continuously differentiable functions, 
we may write Taylor expansions
\begin{align*}
    r(t) &= r(t_0) + (t-t_0) r'(s_r) 
        = R + (t-t_0) r'(s_r) 
        \text{ and} \\
    z(t) &= z(t_0) + (t-t_0) z'(s_z) 
        = Z + (t-t_0) z'(s_z)
\end{align*}
for some values $s_r,s_z$ between
$t_0$ and $t$, which depend on $t$ and which tend to $t_0$ 
as $t \to t_0$. 
It immediately follows that, for fixed $t$, we have
\[
    \frac{Z-z(t)}{(R-r(t))^2 + (Z-z(t))^2} = 
    \frac{-(t-t_0)z'(s_z)}{(t-t_0)^2 (r'(s_r)^2 + z'(s_z)^2) } = 
    \left(\frac 1{t-t_0} \right)
    \left(\frac{-z'(s_z)}{r'(s_r)^2 + z'(s_z)^2} \right).
\]
As long as $z'(t_0) \ne 0$, 
we obtain the answer to our question. The integrand
obeys the asymptotic estimate
\[
    \vec e_x \cdot \vec F_1(t) \sim \frac 1{t-t_0} \qquad 
    \text{as } t \to t_0,
\]
and we conclude that $\|\vec F_1(t)\|$, and hence $\|\vec F(\theta, t)\|$,
are not integrable.

When $z'(t_0) = 0$, we consider the integral of the other component
$\vec e_z \cdot \vec F_1(t)$ and consider whether
\[
    \int_0^L
    \frac {2 \left[\pder\psi z r'(t) - \pder\psi r z'(t) \right]}
        {r(t) \sqrt{\alpha + \beta}} 
    \left( K(k^2) + \frac{r(t)^2 - R^2 - (Z-z(t))^2}
        {\alpha - \beta} E(k^2) \right) \d t
\]
diverges. Using an argument verbatim to the first part of
this proof, we conclude that its behavior is determined
by the singularity of 
\[
    \frac{r(t)^2 - R^2 - (Z-z(t))^2} {\alpha - \beta}
    =
    \frac{r(t)^2 - R^2 - (Z-z(t))^2} {(R-r(t))^2 + (Z-z(t))^2}.
\]
Using the same Taylor expansions, we obtain
\begin{align*}
    \frac{r(t)^2 - R^2 - (Z-z(t))^2} {(R-r(t))^2 + (Z-z(t))^2}
    &= \frac{2 R (t-t_0) r'(s_r) + (t-t_0)^2(r'(s_r)^2 - z'(s_z)^2)} 
        {(t-t_0)^2(r'(s_r)^2 + z'(s_z)^2)} \\
    &= \left( \frac 1{t-t_0} \right) 
        \frac{2 R r'(s_r)}{r'(s_r)^2 + z'(s_z)^2}
        + \frac{r'(s_r)^2 - z'(s_z)^2}{r'(s_r)^2 + z'(s_z)^2}.
\end{align*}
With identical reasoning, and considering that $r'$ and $z'$
cannot simultaneously vanish, we conclude once again that
$\|\vec F(\theta, t)\|$ is not integrable.
\end{proof}

As a result of theorem \ref{th:need_pv}, we conclude that one must use a principal value procedure in order to define all components of $\vec B^{pol}_V(\vec x)$ by the virtual casing principle \eqref{eq:virtual_casing} when $\vec x \in \Gamma$. One can prove that this is possible, i.e., that a limiting procedure yields a finite result, but we omit the lengthy details here.

\end{document}